\documentclass{aCON2e}

\usepackage{float}
\usepackage{ifpdf}
\usepackage{graphicx,amsmath,amssymb,mathrsfs,lineno,cite}
\usepackage{epstopdf}
\usepackage{mathtools}
\usepackage{multicol,multirow}
\usepackage{amsfonts}

\usepackage{caption}
\usepackage[longnamesfirst,sort]{natbib}
\bibpunct[, ]{(}{)}{;}{a}{,}{,}

\newtheorem{mytheorem}{Theorem}
\newtheorem{myproblem}{Problem}

\newtheorem{mydefinition}{Definition}
\newtheorem{mylemma}{Lemma}
\newtheorem{myremark}{Remark}
\newtheorem{myassumption}{Assumption}
\newtheorem{myalgorithm}{Algorithm}

\makeatletter

\newcommand{\Rmnum}[1]{\expandafter\@slowromancap\romannumeral #1@}
\makeatother

\newcommand{\tabincell}[2]{\begin{tabular}{@{}#1@{}}#2\end{tabular}}

\begin{document}

\title{Mixed LQG and $H_\infty$ Coherent Feedback Control for Linear Quantum Systems}

\author{
\name{Lei Cui, Zhiyuan Dong, Guofeng Zhang$^{\ast}$\thanks{$^\ast$Corresponding author. Email: guofeng.zhang@polyu.edu.hk} and Heung Wing Joseph Lee}
\affil{Department of Applied Mathematics, The Hong Kong Polytechnic University, Hung Hom, Kowloon, Hong Kong, China}
}

\maketitle

\begin{abstract}
The purpose of this paper is to study the mixed linear quadratic Gaussian (LQG) and $H_\infty$ optimal control problem for linear quantum stochastic systems, where the controller itself is also a quantum system, often referred to as ``coherent feedback controller''. A lower bound of the LQG control is proved. Then two different methods, rank constrained linear matrix inequality (LMI) method and genetic algorithm are proposed for controller design. A passive system (cavity) and a non-passive one (degenerate parametric amplifier, DPA) demonstrate the effectiveness of these two proposed algorithms.
\end{abstract}


\section{Introduction}

With the rapid development of quantum technology in recent years, more and more researchers are paying attention to quantum control systems, which are an important part in quantum information science. On the other hand, it is found that many methodologies in classical (namely non-quantum) control theory, can be extended into the quantum regime \citep{BHJ2007,DJ1999,DHJ2000,HM2013,JNP2008,WJ2015,ZLH2012}. Meanwhile, quantum control has its special features absent in the classical case, see e.g. \citet{WM2010}, \citet{WNZ2013}, \citet{ZJ2011} and \citet{ZJ2012}. For example, a controller in a quantum feedback control system may be classical, quantum or even mixed quantum/classical \citep{JNP2008}. Generally speaking, in the physics literature, the feedback control problem in which the designed controller is also a fully quantum system is often named as ``coherent feedback control''.

Optimal control, as a vital concept in classical control theory \citep{ZDG1996}, has been widely studied. $H_2$ and $H_{\infty}$ control are the two foremost control methods in classical control theory, which aim to minimize cost functions with specific forms from exogenous inputs (disturbances or noises) to pertinent performance outputs. When the disturbances and measurement noises are Gaussian stochastic processes with known power spectral densities, and the objective is a quadratic performance criterion, then the problem of minimizing this quadratic cost function of linear systems is named as LQG control problem, which has been proved to be equivalent to an $H_2$ optimal control problem \citep{ZDG1996}. On the other hand, $H_{\infty}$ control problem mainly concerns the robustness of a system to parameter uncertainty or external disturbance signals, and a controller to be designed should make the closed-loop system stable, meanwhile minimizing the influence of disturbances or system uncertainties on the system performance in terms of the $H_{\infty}$ norm of a certain transfer function. Furthermore, the mixed LQG (or $H_2$) and $H_{\infty}$ control problem for classical systems has been studied widely during the last three decades. When the control system is subject to both white noises and signals with bounded power, not only optimal performance (measured in $H_2$ norm) but also robustness specifications (in terms of an $H_{\infty}$ constraint) should be taken into account, which is one of the main motivations for considering the mixed control problem \citep{ZGB1994}; see also  \citet{CZ2003}, \citet{DZG1994}, \citet{KR1991}, \citet{NA2004}, \citet{QSY2015}, \citet{ZDG1996} and \citet{ZGB1994} and the references therein.

Very recently, researchers have turned to consider the optimal control problem of quantum systems. For instance, $H_{\infty}$ control of linear quantum stochastic systems is investigated in \citet{JNP2008}, three different types of controllers are designed. \citet{NJP2009} proposes a method for quantum LQG control, for which the designed controller is also a fully quantum system. In \citet{ZJ2011}, direct coupling and indirect coupling for quantum linear systems have been discussed.  It is shown in \citet{ZLH2012} that  phase shifters and ideal squeezers can be used in feedback loop for better control performance. Nevertheless, all of above papers mainly focus on the vacuum inputs, while the authors in \citet{HM2013} concern not only the vacuum case, but also the thermal input. They also discussed how to design both classical and non-classical controllers for LQG control problem. Besides, because of non-linear and non-convex constraints in the coherent quantum controller synthesis, \citet{HP2015} uses a differential evolution algorithm to construct an optimal linear coherent quantum controller. Notwithstanding the above research, to our best knowledge, there is little research on the mixed LQG and $H_{\infty}$ coherent control problem for linear quantum systems, except \citet{BZL2015}.

Similar to the classical case,  in  mixed LQG and $H_{\infty}$ quantum coherent control, LQG and $H_{\infty}$ performances are not independent. Moreover, because the controller to  be designed is another quantum-mechanical system, it has to satisfy additional constraints, which are called ``physical realizability conditions'' \citep{JNP2008,ZJ2012}. For more details, please refer to Section \ref{problem_formulation}.

The contribution of the paper is three-fold. Firstly, a {\it mixed} LQG and $H_{\infty}$ coherent feedback control problem  has been studied, while most of the present literatures (except the conference paper \citet{BZL2015}, by one of the authors) only focus on LQG or $H_{\infty}$ control problem separately. For a typical quantum optical system, there exist quantum white noise as well as finite energy signals (like lasers), while quantum white noise can be dealt with LQG control, finite energy disturbance can better handled by $H_\infty$ control. As a result, it is important to study the mixed control problem. Secondly, we extend Theorem 4.1 in \citet{ZLH2012}, and prove a general result for the lower bound of LQG index. Finally, we propose a genetic-algorithm-based method to design a coherent controller for this mixed problem. In contrast to the numerical algorithm proposed in the earlier conference paper \citep{BZL2015} by one of the authors, the new algorithm is much simpler and is able to produce better results, as clearly demonstrated by numerical studies.

The organization of the paper is as follows. In Section \ref{quantum system}, quantum linear systems are briefly discussed. Section \ref{problem_formulation} formulates the mixed LQG and $H_{\infty}$ coherent feedback control problem. Two algorithms, rank constrained LMI method and genetic algorithm, are proposed in Section \ref{method}. Section \ref{examples} presents numerical studies to demonstrate the proposed algorithms.  Section \ref{conclusion} concludes the paper.

\textbf{Notation}. Let $i=\sqrt{-1}$ be the imaginary unit. $F$ denotes a real skew symmetric $2\times 2$ matrix $F=[0\ 1;-1\ 0]$. Then define a real antisymmetric matrix $\Theta$ with components $\Theta_{jk}$ is \emph{canonical}, which means $\Theta=diag(F,F,...,F)$. Given a column vector of operators $x=[x_1\ \cdots\ x_m]^T$ where $m$ is a positive integer, define $x^{\#}=[x_1^*\ \cdots\ x_m^*]^T$, where the asterisk $*$ indicates Hilbert space adjoint or complex conjugation. Furthermore, define the doubled-up column vector to be $\breve{x}=[x^T\ (x^{\#})^T]^T$, and the matrix case can be defined analogously. Given two matrices $U,V\in\mathbb{C}^{r\times k}$, a doubled-up matrix $\Delta(U,V)$ is defined as $\Delta(U,V)\ $$\coloneqq$$\ [U\ V;V^{\#}\ U^{\#}]$. Let $I_N$ be an identity matrix of dimension $N$, and define $J_N=diag(I_N,-I_N)$, where the ``diag'' notation indicates a block diagonal matrix assembled from the given entries. Then for a matrix $X\in\mathbb{C}^{2N\times 2M}$, define $X^{\flat}\ $$\coloneqq$$\ J_M X^{\dag} J_N$. Finally, the symbol $[\text{ },\text{ }]$ is defined for commutator $[A,B] \coloneqq AB-BA$.

\section{Linear quantum systems}\label{quantum system}

\subsection{Open linear quantum systems} \label{subsec:system}

An open linear quantum system $G$ consists of $N$ quantum harmonic oscillators $a=[a_{1}\ \cdots\ a_{N}]^T$ interacting with $N_w$-channel quantum fields. Here $a_j$ is the \emph{annihilation operator} of the $j$th quantum harmonic oscillator and $a_j^*$ is the \emph{creation operator}, they satisfy canonical commutation relations (CCR): $[a_j,a_k^*]=\delta_{jk}$, and $[a_j,a_k]=[a_j^*,a_k^*]=0$ ($j,k=1,...,N$). Such a linear quantum system can be specified by a triple of physical parameters $(S,L,H)$ \citep{HP1984}.

In this triple, $S$ is a unitary scattering matrix of dimension $N_w$. $L$ is a vector of coupling operators defined by
\begin{equation}
L=C_{-}a+C_{+}a^{\#}
\end{equation}
where $C_{-}$ and $C_{+}$ $\in\mathbb{C}^{N_w\times N}$. $H$ is the Hamiltonian describing the self-energy of the system, satisfying
\begin{equation}
H=\frac{1}{2} \breve{a}^{\dag} \Delta(\Omega_{-},\Omega_{+}) \breve{a}
\end{equation}
where $\Omega_{-}$ and $\Omega_{+}$ $\in\mathbb{C}^{N\times N}$ with $\Omega_{-}=\Omega_{-}^{\dag}$ and $\Omega_{+}=\Omega_{+}^T$.

The \emph{annihilation-creation representation} for linear quantum stochastic systems can be written as the following quantum stochastic differential equations (QSDEs)
\begin{equation}\label{a-c form}\begin{split}
d\breve{a}(t)&=\breve{A}\breve{a}(t)dt+\breve{B}d\breve{b}_{in}(t),\quad \breve{a}(0)=\breve{a}_0 \\
d\breve{y}(t)&=\breve{C}\breve{a}(t)dt+\breve{D}d\breve{b}_{in}(t),
\end{split}\end{equation}
where the correspondences between system matrices $(\breve{A},\breve{B},\breve{C},\breve{D})$ and parameters $(S,L,H)$ are as follows
\begin{equation}\begin{split}
\breve{A}&=-\frac{1}{2}\breve{C}^{\flat}\breve{C}-iJ_{N}\Delta(\Omega_{-},\Omega_{+}),\quad \breve{B}=-\breve{C}^{\flat}\Delta(S,0), \\
\breve{C}&=\Delta(C_{-},C_{+}),\quad \breve{D}=\Delta(S,0).
\end{split}\end{equation}

\subsection{Quadrature representation of linear quantum systems}

In addition to annihilation-creation representation, there is an alternative form, \emph{amplitude-phase quadrature representation}, where all the operators are observable (self-adjoint operators) and all corresponding matrices are real, so this form is more convenient for standard matrix analysis software packages and programmes \citep{BZL2015,ZLH2012}.

Firstly, define a unitary matrix
\begin{equation}
\Lambda_n=\frac{1}{\sqrt{2}}\begin{bmatrix}
I & I \\
-iI & iI
\end{bmatrix}_{n\times n}
\end{equation}
and denote $q_j=(a_j+a_j^*)/\sqrt{2}$ as the \emph{real quadrature}, and $p_j=(-i a_j+i a_j^*)/\sqrt{2}$ as the \emph{imaginary} or \emph{phase quadrature}. It is easy to show these two quadratures also satisfy the CCR $[q_j,p_k]=i \delta_{jk}$ and $[q_j,q_k]=[p_j,p_k]=0$ ($j,k=1,...,N$).

By defining a coordinate transform
\begin{equation}
x \coloneqq \Lambda_n \breve{a},\quad
w \coloneqq \Lambda_{n_w} \breve{b}_{in},\quad
y \coloneqq \Lambda_{n_y} \breve{y},
\end{equation}
we could get
\begin{equation}\label{q form}\begin{split}
dx(t)&=Ax(t)dt+Bdw(t),\quad x(0)=x_0 \\
dy(t)&=Cx(t)dt+Ddw(t),
\end{split}\end{equation}
where $n=2N$, $n_w=2 N_w$, $n_y=2 N_y$ are positive even integers, and $x(t)=[q_1(t)\ \cdots\ q_N(t)\ p_1(t)\ \cdots\ p_N(t)]^T$ is the vector of system variables, $w(t)=[w_1(t)\ \cdots\ w_{n_w}(t)]^T$ is the vector of input signals, including control input signals, noises and disturbances, $y(t)=[y_1(t)\ \cdots\ y_{n_y}(t)]^T$ is the vector of outputs. $A$, $B$, $C$ and $D$ are real matrices in $\mathbb{R}^{n\times n}$, $\mathbb{R}^{n\times n_w}$, $\mathbb{R}^{n_y\times n}$ and $\mathbb{R}^{n_y\times n_w}$, respectively. The correspondences between these coefficient matrices of two different representations are
\begin{equation}\label{trans matrix}\begin{split}
A &=\Lambda_n \breve{A} \Lambda_n^{\dag},\quad B=\Lambda_n \breve{B} \Lambda_{n_w}^{\dag}, \\
C &=\Lambda_{n_y} \breve{C} \Lambda_n^{\dag},\quad D=\Lambda_{n_y} \breve{D} \Lambda_{n_w}^{\dag}.
\end{split}\end{equation}

\begin{myremark}
For simplicity in calculation, we usually do a simple linear transformation to obtain $x(t)=[x_1(t)\ \cdots\ x_n(t)]^T=[q_1(t)\ p_1(t)\ \cdots\ q_N(t)\ p_N(t)]^T$, and similarly in $w(t)$, $y(t)$ and corresponding matrices \citep{ZLH2012}. In the rest of this paper, we only focus on the quadrature form after the linear transformation.
\end{myremark}

\begin{myassumption}\label{assum1}
Without loss of generality, we give some assumptions for quantum systems \citep{BZL2015,NJP2009}.
\begin{enumerate}
\item The initial system variable $x(0)=x_0$ is Gaussian.
\item The vector of inputs $w(t)$ could be decomposed as $dw(t)=\beta_w (t)dt+d\tilde{w}(t)$, where $\beta_w (t)$ is a self-adjoint adapted process, $\tilde{w}(t)$ is the noise part of $w(t)$, and satisfies $d\tilde{w}(t) d\tilde{w}^T (t)=F_{\tilde{w}}dt$, where $F_{\tilde{w}}$ is a nonnegative Hermitian matrix. In quantum optics, $\tilde{w}(t)$ is quantum white noise, and $\beta_w (t)$ is the signal, which in many cases is $L_2$ integrable.

\item The components of $\beta_w (t)$ commute with those of $d\tilde{w}(t)$ and also those of $x(t)$ for all $t\geq 0$.
 \end{enumerate}
\end{myassumption}

\subsection{Physical realizability conditions of linear QSDEs}

The QSDEs (\ref{q form}) may not necessarily represent the dynamics of a meaningful physical system, because quantum mechanics demands physical quantum systems to evolve in a unitary manner. This implies the preservation of canonical commutation relations $x(t)x^T (t)-(x(t)x^T (t))^T=i \Theta$ for all $t\geq 0$, and also another constraint related to the output signal. These constraints are formulated as physically realizability of quantum linear systems in \citet{JNP2008}.

A linear noncommutative stochastic system of quadrature form (\ref{q form}) is physically realizable if and only if
\begin{subequations}\label{phy real}\begin{align}
iA\Theta+i\Theta A^T+B T_{\tilde{w}} B^T=0&, \\
B\begin{bmatrix}
I_{n_y \times n_y} \\
0_{(n_w - n_y)\times n_y}
\end{bmatrix}=\Theta C^T diag_{N_y}(F)&, \\
D=\begin{bmatrix}
I_{n_y \times n_y} & 0_{n_y \times (n_w - n_y)}
\end{bmatrix}&,
\end{align}\end{subequations}
where the first equation determines the Hamiltonian and coupling operators, and the others relate to the required form of the output equation.

\subsection{Direct coupling}\label{add com}

There are also some additional components and relations in quantum systems, such as direct coupling, phase shifter, ideal squeezer, etc. Interested readers could refer to e.g. \citet{ZJ2011,ZJ2012,ZLH2012}. Depending on the need of this paper, we just briefly introduce the \emph{direct coupling}.

In quantum mechanics, two independent systems $G_1$ and $G_2$ may interact by exchanging energy directly. This energy exchange can be described by an interaction Hamiltonian $H_{int}$. In this case, it is said that these two systems are directly coupled. When they are expressed in annihilation-creation operator form, such as:
\begin{equation*}\begin{split}
d\breve{a}_1 (t) &=\breve{A}_1 \breve{a}_1 (t)dt+\breve{B}_{12} \breve{a}_2 (t)dt, \\
d\breve{a}_2 (t) &=\breve{A}_2 \breve{a}_2 (t)dt+\breve{B}_{21} \breve{a}_1 (t)dt,
\end{split}\end{equation*}
where the subscript $1$ means that corresponding terms belong to the system $G_1$, and similar for subscript $2$. $B_{12}$ and $B_{21}$ denote the direct coupling between two systems, and satisfy the relations as follows
\begin{equation*}\begin{split}
B_{12} &=-\Delta(K_{-},K_{+})^{\flat}, \\
B_{21} &=-B_{12}^{\flat}=\Delta(K_{-},K_{+}),
\end{split}\end{equation*}
where $K_- $ and $K_+$ are arbitrary constant matrices of appropriate dimensions.

\begin{mydefinition}
For a quantum linear system in the annihilation-creation operator form which is defined by parameters $(C_{-},C_{+},\Omega_{-},\Omega_{+},K_{-},K_{+})$, there will have the following classifications:
\begin{enumerate}
\item If all ``plus'' parameters (i.e. $C_{+}$, $\Omega_{+}$ and $K_{+}$) are equal to 0, the system is called a passive system;
\item Otherwise, it is called a non-passive system.
\end{enumerate}
\end{mydefinition}

Examples for these two different systems are given in Section \ref{examples}.

\section{Synthesis of mixed LQG and $H_{\infty}$ coherent feedback control problem}\label{problem_formulation}

In this section, we firstly formulate the QSDEs for the closed-loop system, in which both plant and controller are quantum systems, as well as the specific physical realizability conditions. Then $H_\infty$ and LQG control problems are discussed.

\subsection{Composite plant-controller system}\label{our system}

\begin{figure}
\begin{center}
{\includegraphics[scale=.9]{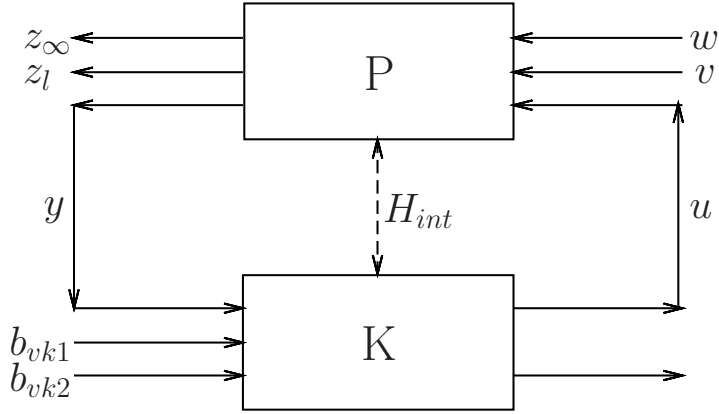}}
\caption{Schematic of the closed-loop plant-controller system.}
\label{sys fig}
\end{center}
\end{figure}

Consider the closed-loop system as shown in Figure \ref{sys fig}. The quantum plant $P$ is described by QSDEs in quadrature form \citep{BZL2015}
\begin{equation}\label{plant}\begin{split}
dx(t) &=Ax(t)dt+B_0 dv(t)+B_1 dw(t)+B_2 du(t), \\
dy(t) &=C_2 x(t)dt+D_{20}dv(t)+D_{21}dw(t), \\
dz_{\infty}(t) &=C_1 x(t)dt+D_{12}du(t), \\
z_l (t) &=C_z x(t)+D_z \beta_u (t),
\end{split}\end{equation}
where $A$, $B_0$, $B_1$, $B_2$, $C_2$, $D_{20}$, $D_{21}$, $C_1$, $D_{12}$, $C_z$ and $D_z$ are real matrices in $\mathbb{R}^{n\times n}$, $\mathbb{R}^{n\times n_v}$, $\mathbb{R}^{n\times n_w}$, $\mathbb{R}^{n\times n_u}$, $\mathbb{R}^{n_y \times n}$, $\mathbb{R}^{n_y \times n_v}$, $\mathbb{R}^{n_y \times n_w}$, $\mathbb{R}^{n_{\infty} \times n}$, $\mathbb{R}^{n_{\infty} \times n_u}$, $\mathbb{R}^{n_l \times n}$, $\mathbb{R}^{n_l \times n_u}$ respectively, and $n$, $n_v$, $n_w$, $n_u$, $n_y$, $n_{\infty}$, $n_l$ are positive integers. $x(t)=[x_1 (t)\ \cdots\ x_n (t)]^T$ is the vector of self-adjoint possibly noncommutative system variables; $u(t)=[u_1 (t)\ \cdots\ u_{n_u}(t)]^T$ is the controlled input; $v(t)=[v_1 (t)\ \cdots\ v_{n_v}(t)]^T$ and $w(t)=[w_1 (t)\ \cdots\ w_{n_w}(t)]^T$ are other inputs. $z_{\infty}(t)=[z_{\infty_1}(t)\ \cdots\ z_{\infty_{n_{\infty}}}(t)]^T$ and $z_l (t)=[z_{l_1}(t)\ \cdots\ z_{l_{n_l}}(t)]^T$ are controlled outputs which are referred to as $H_{\infty}$ and LQG performance, respectively.

The purpose is to design a coherent feedback controller $K$ to minimize the LQG norm and the $H_{\infty}$ norm of closed-loop system simultaneously, and $K$ has the following form
\begin{equation}\label{controller}\begin{split}
d\xi (t) &=A_k \xi (t)dt+B_{k1}db_{vk1}(t)+B_{k2}db_{vk2}(t)+B_{k3}dy(t), \\
du(t) &=C_k \xi (t)dt+db_{vk1}(t),
\end{split}\end{equation}
where $\xi (t)=[\xi_1 (t)\ \cdots\ \xi_{n_k}(t)]^T$ is a vector of self-adjoint variables, and matrices $A_k$, $B_{k1}$, $B_{k2}$, $B_{k3}$, $C_k$ have appropriate dimensions.

\begin{myassumption}\label{assum2}
Similarly with Assumption \ref{assum1}, we give additional assumptions for the specific plant and controller which we consider.
\begin{enumerate}
\item The inputs $w(t)$ and $u(t)$ also have the decompositions: $dw(t)=\beta_w (t)dt+d\tilde{w}(t)$, $du(t)=\beta_u (t)dt+d\tilde{u}(t)$, where the meanings of $\beta_{\ast}$ and $\tilde{\ast}$ are similar as those in Assumption \ref{assum1};
\item The controller state variable $\xi (t)$ has the same order as the plant state variable $x(t)$;
\item $v(t)$, $\tilde{w}(t)$, $b_{vk1}(t)$ and $b_{vk2}(t)$ are independent quantum noises in vacuum state;
\item The initial plant state and controller state satisfy relations: $x(0)x^T (0)-(x(0)x^T (0))^T=i \Theta$, $\xi (0) \xi^T (0)-(\xi (0) \xi^T (0))^T=i \Theta_k$, $x(0)\xi^T (0)-(\xi (0) x^T (0))^T=0$.
\end{enumerate}
\end{myassumption}

By the identification $\beta_u (t) \equiv C_k \xi (t)$ and $\tilde{u}(t) \equiv b_{vk1}(t)$, the closed-loop system is obtained as
\begin{equation}\label{cl-loop system}\begin{split}
d\eta(t) &=M\eta(t)dt+Nd\tilde{w}_{cl}(t)+H\beta_w (t)dt, \\
dz_{\infty}(t) &=\Gamma \eta(t)dt+\Pi d\tilde{w}_{cl}(t), \\
z_l (t) &=\Psi \eta(t),
\end{split}\end{equation}
where $\eta(t)=[x^T (t)\ \xi^T (t)]^T$ denotes the state of the closed-loop system, $\beta_w (t)$ is the disturbance, while $\tilde{w}_{cl}(t)=[v^T (t)\ \tilde{w}^T (t)\ b_{vk1}^T (t)\ b_{vk2}^T (t)]^T$ contains all white noises, and coefficient matrices are shown as follows
\begin{equation*}\begin{split}
M &=\begin{bmatrix}
A & B_2 C_k \\
B_{k3}C_2 & A_k
\end{bmatrix}, \\
N &=\begin{bmatrix}
B_0 & B_1 & B_2 & 0 \\
B_{k3}D_{20} & B_{k3}D_{21} & B_{k1} & B_{k2}
\end{bmatrix}, \\
H &=\begin{bmatrix}
B_1 \\
B_{k3}D_{21}
\end{bmatrix},
\quad
\Gamma =\begin{bmatrix}
C_1 & D_{12}C_k
\end{bmatrix}, \\
\Pi &=\begin{bmatrix}
0 & 0 & D_{12} & 0
\end{bmatrix},
\quad
\Psi =\begin{bmatrix}
C_z & D_z C_k
\end{bmatrix}.
\end{split}
\end{equation*}

\subsection{Physical realizability conditions}

For the plant $P$ introduced in the previous section, we want to design a controller $K$ which is also a quantum system. Hence from \citet{JNP2008} and \citet{ZLH2012}, the equations (\ref{controller}) should also satisfy the following physical realizability conditions
\begin{subequations}\label{our phy thm}\begin{align}
A_k \Theta_k+\Theta_k A_k^T &+B_{k1}diag_{n_{vk1} /2}(F)B_{k1}^T \nonumber \\
&+B_{k2}diag_{n_{vk2} /2}(F)B_{k2}^T \nonumber \\
&+B_{k3}diag_{n_{vk3} /2}(F)B_{k3}^T=0, \\
B_{k1} &=\Theta_k C_k^T diag_{n_u /2}(F).
\end{align}\end{subequations}

\subsection{LQG control problem}\label{LQG analysis}

For the closed-loop system (\ref{cl-loop system}), we associate a quadratic performance index
\begin{equation}\label{qua index}
J(t_f)=\int_0^{t_f} \langle z_l^T (t) z_l (t) \rangle dt,
\end{equation}
where the notation $\langle \cdot \rangle$ is standard and refers to as quantum expectation \citep{Mer1998}.

\begin{myremark}
In classical control, $\int_0^\infty (x(t)^T P x(t) + u(t)^T Q u(t)) dt$ is the standard form for LQG performance index, where $x$ is the system variable and $u$ is the control input.  However, things are more complicated in the quantum regime. By Eq. (\ref{controller}), we can see that $u(t)$ is a function of both $\xi(t)$ (the controller variable) and $b_{vk1}(t)$ (input quantum white noise).  If we use $u(t)$ in Eq. (\ref{controller}) directly, then there will be quantum white noise in the LQG performance index, which yields an unbounded LQG control performance.  On the other hand, by Eq. (\ref{qua index}), the LQG performance index is a function of $x(t)$ (the system variable) and $\xi(t)$ (the controller variable). This is the appropriate counterpart of the classical case.
\end{myremark}

Generally, we always focus on the infinite horizon case $t_f \to \infty$. Therefore, as in \citet{NJP2009}, assume that $M$ is asymptotically stable, by standard analysis methods, we have the infinite-horizon LQG performance index as
\begin{equation} \label{eq:lqg}
J_{\infty}=\lim_{t_f \to \infty} \frac{1}{t_f} \int_0^{t_f} \langle z_l^T (t) z_l (t) \rangle dt=\mathrm{Tr}(\Psi P \Psi^T),
\end{equation}
where $P$ is the unique symmetric positive definite solution of the Lyapunov equation
\begin{equation}\label{lya eq}
M P+P M^T +\frac{1}{2} N N^T=0.
\end{equation}

\begin{myproblem}\label{ori LQG}
The LQG coherent feedback control problem is to find a quantum controller $K$ of equations (\ref{controller}) that minimizes the LQG performance index $J_{\infty}=\mathrm{Tr}(\Psi P \Psi^T)$. Here $P$ is the unique solution of equation (\ref{lya eq}), and coefficient matrices of controller satisfy constraints (\ref{our phy thm}).
\end{myproblem}

When considering minimizing LQG performance index, firstly we want to know the minimum. But for general case, it is too complicated to get the theoretical result, so we choose the orders of plant and controller to be 2.  In this case, because $C_z^T C_z$ and $D_z^T D_z$ are both 2-by-2 positive semi-definite real matrices, we denote
\begin{equation*}
C_z^T C_z=\begin{bmatrix}
c_1 & c_2 \\
c_2 & c_3
\end{bmatrix},
D_z^T D_z=\begin{bmatrix}
d_1 & d_2 \\
d_2 & d_3
\end{bmatrix},
C_k=\begin{bmatrix}
c_{k1} & c_{k2} \\
c_{k3} & c_{k4}
\end{bmatrix},
\end{equation*}
where all parameters in these matrices are real scalars.

In analogy to Theorem 4.1 in \citet{ZLH2012}, we have the following result.

\begin{mytheorem}\label{lower bound LQG}(The lower bound of LQG index)
Assume that both the plant and the controller defined in Section \ref{our system} are in the ground state, then LQG performance index
\begin{equation*}
J_{\infty} \geq \frac{c_1 +c_3}{2}+d_2 (c_{k1} c_{k3}+c_{k2} c_{k4}),
\end{equation*}
where $c_{\ast}$ and $d_{\ast}$ come from the matrices above.

\end{mytheorem}

\begin{proof}
Since $z_l =C_z x+D_z \beta_u =C_z x +D_z C_k \xi$, we could easily get
\begin{equation}\label{zl eq}\begin{split}
\langle z_l^T z_l \rangle &=\langle (C_z x+D_z C_k \xi)^T (C_z x+D_z C_k \xi) \rangle \\
&=\langle x^T C_z^T C_z x \rangle + \langle \xi^T C_k^T D_z^T D_z C_k \xi \rangle \\
&\phantom{=}+\langle x^T C_z^T D_z C_k \xi \rangle + \langle \xi^T C_k^T D_z^T C_z x \rangle,
\end{split}\end{equation}
where
\begin{equation*}
x=\begin{bmatrix}
q \\
p
\end{bmatrix}
=\frac{1}{\sqrt{2}}\begin{bmatrix}
1 & 1 \\
-i & i
\end{bmatrix}\begin{bmatrix}
a \\
a^*
\end{bmatrix},
\xi =\begin{bmatrix}
q_k \\
p_k
\end{bmatrix}
=\frac{1}{\sqrt{2}}\begin{bmatrix}
1 & 1 \\
-i & i
\end{bmatrix}\begin{bmatrix}
a_k \\
a_k^*
\end{bmatrix}.
\end{equation*}

Then we have
\begin{equation}\label{first term}\begin{split}
&\langle x^T C_z^T C_z x \rangle =\frac{1}{2} \langle [a\ a^*] \begin{bmatrix}
1 & -i \\
1 & i
\end{bmatrix}\begin{bmatrix}
c_1 & c_2 \\
c_2 & c_3
\end{bmatrix}\begin{bmatrix}
1 & 1 \\
-i & i
\end{bmatrix}\begin{bmatrix}
a \\
a^*
\end{bmatrix} \rangle \\
&=\frac{1}{2} \langle [a\ a^*] \begin{bmatrix}
c_1 -c_3 -2ic_2 & c_1 +c_3 \\
c_1 +c_3 & c_1 -c_3 +2ic_2
\end{bmatrix}\begin{bmatrix}
a \\
a^*
\end{bmatrix} \rangle \\
&=\frac{1}{2} \langle (c_1 +c_3)a^* a+(c_1 +c_3)a a^* \\
&\phantom{=}+(c_1 -c_3 -2ic_2)aa+(c_1 -c_3 +2ic_2)a^* a^*] \rangle \\
&=\langle (c_1 +c_3)a^* a +\frac{c_1 +c_3}{2} \rangle,
\end{split}\end{equation}
where the last equality follows from our assumption that the plant are in the ground state, and $[a,a^*]=1 \Rightarrow a a^* =1+a^* a$. The second term of equation (\ref{zl eq}) becomes
\begin{equation}\begin{split}
&\phantom{=} \langle \xi^T C_k^T D_z^T D_z C_k \xi \rangle \\
&=\langle [q_k\ p_k] \begin{bmatrix}
c_{k1} & c_{k3} \\
c_{k2} & c_{k4}
\end{bmatrix}\begin{bmatrix}
d_1 & d_2 \\
d_2 & d_3
\end{bmatrix}\begin{bmatrix}
c_{k1} & c_{k2} \\
c_{k3} & c_{k4}
\end{bmatrix}\begin{bmatrix}
q_k \\
p_k
\end{bmatrix} \rangle \\
&=\langle [q_k\ p_k] \begin{bmatrix}
e_1 & e_2 \\
e_2 & e_3
\end{bmatrix}\begin{bmatrix}
q_k \\
p_k
\end{bmatrix} \rangle \\
&=\langle e_1 q_k^2 +e_3 p_k^2 +e_2 (q_k p_k +p_k q_k) \rangle,
\end{split}\end{equation}
where $e_1 =d_1 c_{k1}^2 +d_3 c_{k3}^2 +2d_2 c_{k1}c_{k3}$, $e_3 =d_1 c_{k2}^2 +d_3 c_{k4}^2 +2d_2 c_{k2}c_{k4}$, $e_2 =d_1 c_{k1}c_{k2}+d_3 c_{k3}c_{k4}+d_2 (c_{k1}c_{k4}+c_{k2}c_{k3})$.

While $q_k =\frac{a_k +a_k^*}{\sqrt{2}}$ and $p_k =\frac{-i a_k +i a_k^*}{\sqrt{2}}$, we get
\begin{equation*}\begin{split}
q_k^2 &=\frac{1}{2} [a_k^2 +(a_k^*)^2 +2a_k^* a_k +1], \\
p_k^2 &=-\frac{1}{2} [a_k^2 +(a_k^*)^2 -2a_k^* a_k -1], \\
q_k p_k +p_k q_k &= -i [a_k^2 -(a_k^*)^2],
\end{split}\end{equation*}
and
\begin{equation}\label{second term}\begin{split}
&\phantom{=} \langle \xi^T C_k^T D_z^T D_z C_k \xi \rangle \\
&=\langle \frac{e_1}{2} [a_k^2 +(a_k^*)^2 +2a_k^* a_k +1] \\
&\phantom{=} -\frac{e_3}{2} [a_k^2 +(a_k^*)^2 -2a_k^* a_k -1]-e_2 i [a_k^2 -(a_k^*)^2] \rangle.
\end{split}\end{equation}

Since both the plant and the controller are in the ground state, all terms containing $a$, $a^*$, $a_k$ and $a_k^*$ are 0; and the plant state $x$ commutes with the controller state $\xi$, so the third and fourth terms of equation (\ref{zl eq}) are also 0. By substituting (\ref{first term}) and (\ref{second term}) into (\ref{zl eq}), we obtain the result of $\langle z_l^T z_l \rangle$:
\begin{equation}\begin{split}
\langle z_l^T z_l \rangle &=\frac{c_1 +c_3}{2}+\frac{e_1 +e_3}{2} \\
&=\frac{d_1 (c_{k1}^2 +c_{k2}^2) +d_3 (c_{k3}^2 +c_{k4}^2)+2d_2 (c_{k1} c_{k3}+c_{k2} c_{k4})}{2} \\
&\phantom{=}+\frac{c_1 +c_3}{2}.
\end{split}\end{equation}

Consequently, all square terms are not less than 0, so $J_{\infty} \geq \frac{c_1 +c_3}{2}+d_2 (c_{k1} c_{k3}+c_{k2} c_{k4})$. The proof is completed.
\end{proof}

\begin{myremark}
Sometimes for simplicity, we could choose the coefficient matrix $D_z$ satisfying $d_2=0$, then the bound of LQG index becomes $J_{\infty} \geq \frac{c_1 +c_3}{2}$, which is a constant, independent with the designed controller. This is consistent with the result in \citet{ZLH2012}.
\end{myremark}

Meanwhile, it is easy to see that physical realizability conditions (\ref{our phy thm}) of the coherent controller $K$ are polynomial equality constraints, so they are difficult to solve numerically using general existing optimization algorithms. Hence sometimes we reformulate Problem \ref{ori LQG} into a rank constrained LMI feasibility problem, by letting the LQG performance index $J_{\infty}<\gamma_l$ for a prespecified constant $\gamma_l >0$. This is given by the following result.

\begin{mylemma}(Relaxed LQG problem \citep{NJP2009})
Given $\Theta_k$ and $\gamma_l >0$, if there exist symmetric matrix $P_L=P^{-1}$, $Q$ and coefficient matrices of controller such that physical realizability constraints (\ref{our phy thm}) and following inequality constraints
\begin{equation}\label{nLMI LQG}\begin{split}
\begin{bmatrix}
M^T P_L+P_L M & P_L N \\
N^T P_L & -I
\end{bmatrix} &<0, \\
\begin{bmatrix}
P_L & \Psi^T \\
\Psi & Q
\end{bmatrix} &>0, \\
\mathrm{Tr}(Q) &<\gamma_l
\end{split}\end{equation}
hold, then the LQG coherent feedback control problem admits a coherent feedback controller $K$ of the form (\ref{controller}).
\end{mylemma}

\subsection{$H_{\infty}$ control problem}\label{Hinf analysis}

For linear systems, the $H_{\infty}$ norm can be expressed as follows
\begin{equation}
\lVert T \rVert_{\infty}=\sup_{\omega \in \mathbb{R}}\sigma_{\max}[T(j\omega)]=\sup_{\omega\in \mathbb{R}}\sqrt{\lambda_{\max}(T^* (j\omega)T(j\omega))}
\end{equation}
where $\sigma_{\max}$ is the maximum singular value of a matrix, and $\lambda_{\max}$ is the maximum eigenvalue of a matrix.

Since we consider the $H_{\infty}$ control problem for the closed-loop system (\ref{cl-loop system}), and only $\beta_w$ part contains exogenous signals while the others are all white noises, we interpret $\beta_w \to z_{\infty}$ as the robustness channel for measuring $H_{\infty}$ performance, and our objective to be minimized is
\begin{equation}\begin{split}
\lVert G_{\beta_w \to z_{\infty}} \rVert_{\infty} &=\lVert D_{cl}+C_{cl}(sI-A_{cl})^{-1}B_{cl} \rVert_{\infty} \\
&=\lVert \Gamma (sI-M)^{-1} H \rVert_{\infty}
\end{split}\end{equation}

\begin{myproblem}\label{ori Hinf}
The $H_{\infty}$ coherent feedback control problem is to find a quantum controller $K$ of form (\ref{controller}) that minimizes the $H_{\infty}$ performance index $\lVert G_{\beta_w \to z_{\infty}} \rVert_{\infty}$, while coefficient matrices of controller $A_k$, $B_{k1}$, $B_{k2}$, $B_{k3}$ and $C_k$ satisfy constraints (\ref{our phy thm}) simultaneously.
\end{myproblem}

Similarly to the LQG case, we proceed to relax Problem \ref{ori Hinf} into a rank constrained LMI feasibility problem, i.e. let $\lVert G_{\beta_w \to z_{\infty}} \rVert_{\infty}<\gamma_{\infty}$ for a prespecified constant $\gamma_{\infty} >0$, then we get the following lemma.

\begin{mylemma}(Relaxed $H_{\infty}$ problem \citep{ZJ2011})
Given $\Theta_k$ and $\gamma_{\infty} >0$, if there exist $A_k$, $B_{k1}$, $B_{k2}$, $B_{k3}$, $C_k$ and a symmetric matrix $P_H$ such that physical realizability constraints (\ref{our phy thm}) and following inequality constraints
\begin{equation}\label{nLMI Hinf}\begin{split}
\begin{bmatrix}
M^T P_H+P_H M & P_H H & \Gamma^T \\
H^T P_H & -\gamma_{\infty} I & 0 \\
\Gamma & 0 & -\gamma_{\infty} I
\end{bmatrix} &<0, \\
P_H &>0
\end{split}\end{equation}
hold, then the $H_{\infty}$ coherent feedback control problem admits a coherent feedback controller $K$ of the form (\ref{controller}).
\end{mylemma}

Meanwhile, we also want to know the lower bound of $H_{\infty}$ performance index.  It is in general difficult to derive the minimum value of $H_{\infty}$ index analytically, here we just present a simple example. We begin with the following remark.

\begin{myremark}\label{lower bound Hinf}
By referring to \citet{JNP2008}, there exists an $H_{\infty}$ controller of form (\ref{controller}) for the quantum system (\ref{plant}), if and only if the following pair of algebraic Riccati equations
\begin{equation}\label{riccati eq1}\begin{split}
(A &-B_2 E_1^{-1} D_{12}^T C_1)^T X+X(A -B_2 E_1^{-1} D_{12}^T C_1) \\
&+X(B_1 B_1^T - \gamma_{\infty}^2 B_2 E_1^{-1} B_2^T)X \\
&+\gamma_{\infty}^{-2} C_1^T (I-D_{12} E_1^{-1} D_{12}^T) C_1 =0
\end{split}\end{equation}
and
\begin{equation}\label{riccati eq2}\begin{split}
(A &-B_1 D_{21}^T E_2^{-1} C_2)Y+Y(A -B_1 D_{21}^T E_2^{-1} C_2)^T \\
&+Y(\gamma_{\infty}^{-2} C_1^T C_1 -C_2^T E_2^{-1} C_2)Y \\
&+B_1 (I-D_{21}^T E_2^{-1} D_{21}) B_1^T =0
\end{split}\end{equation}
have positive definite solutions $X$ and $Y$, where $D_{12}^T D_{12}=E_1 >0$, $D_{21} D_{21}^T =E_2 >0$.

We consider a simple example. The system equations are described as:
\begin{equation*}\begin{split}
dx(t) &=-\frac{1}{2} \begin{bmatrix}
0.89 & 0 \\
0 & 0.91
\end{bmatrix} x(t)dt -\sqrt{0.5} \begin{bmatrix}
1 & 0 \\
0 & 1
\end{bmatrix} dv(t) \\
&\phantom{=} -\sqrt{0.2} \begin{bmatrix}
1 & 0 \\
0 & 1
\end{bmatrix} dw(t) -\sqrt{0.2} \begin{bmatrix}
1 & 0 \\
0 & 1
\end{bmatrix} du(t), \\
dy(t) &=\sqrt{0.5} \begin{bmatrix}
1 & 0 \\
0 & 1
\end{bmatrix} x(t)dt +\begin{bmatrix}
1 & 0 \\
0 & 1
\end{bmatrix} dv(t) + \delta \begin{bmatrix}
1 & 0 \\
0 & 1
\end{bmatrix} dw(t), \\
dz_{\infty}(t) &=\sqrt{0.2} \begin{bmatrix}
1 & 0 \\
0 & 1
\end{bmatrix} x(t)dt +\begin{bmatrix}
1 & 0 \\
0 & 1
\end{bmatrix} du(t),
\end{split}\end{equation*}
where $\delta$ is a very small positive real number.

There has no problem to calculate the first Riccati equation (\ref{riccati eq1}). For the second one (\ref{riccati eq2}), denote $Y=\begin{bmatrix}
y_1 & y_2 \\
y_2 & y_3
\end{bmatrix}$, we get
\begin{equation}\label{lb H}
\begin{bmatrix}
(\frac{0.2}{\gamma_{\infty}^2}-\frac{0.5}{\delta^2})(y_1^2 +y_2^2)-(0.89-\frac{2\sqrt{0.1}}{\delta})y_1 & [(\frac{0.2}{\gamma_{\infty}^2}-\frac{0.5}{\delta^2})(y_1 + y_3)-(0.9-\frac{2\sqrt{0.1}}{\delta})]y_2 \\
[(\frac{0.2}{\gamma_{\infty}^2}-\frac{0.5}{\delta^2})(y_1 + y_3)-(0.9-\frac{2\sqrt{0.1}}{\delta})]y_2 & (\frac{0.2}{\gamma_{\infty}^2}-\frac{0.5}{\delta^2})(y_2^2 +y_3^2)-(0.91-\frac{2\sqrt{0.1}}{\delta})y_3
\end{bmatrix}=0.
\end{equation}
Notice that, since $\delta$ is very small, $0.89-\frac{2\sqrt{0.1}}{\delta}$, $0.9-\frac{2\sqrt{0.1}}{\delta}$ and $0.91-\frac{2\sqrt{0.1}}{\delta}$ are negative.

From the (1,2) term, we make a classification: $y_2 = 0$ or $y_2 \neq 0$.
\begin{enumerate}
\item $y_2=0$: Since (1,1) and (2,2) terms are 0, we get
\begin{equation*}
y_1=0\quad or\quad y_1=\frac{0.89-\frac{2\sqrt{0.1}}{\delta}}{\frac{0.2}{\gamma_{\infty}^2}-\frac{0.5}{\delta^2}},
\end{equation*}
\begin{equation*}
y_3=0\quad or\quad y_3=\frac{0.91-\frac{2\sqrt{0.1}}{\delta}}{\frac{0.2}{\gamma_{\infty}^2}-\frac{0.5}{\delta^2}}.
\end{equation*}
\item $y_2 \neq 0$: From the (1,2) term we get
\begin{equation*}
y_1 + y_3 =\frac{0.9-\frac{2\sqrt{0.1}}{\delta}}{\frac{0.2}{\gamma_{\infty}^2}-\frac{0.5}{\delta^2}}.
\end{equation*}
After doing the calculation that the (1,1) term minus the (2,2) term, and substituting $y_1 + y_3$ into it, we get
\begin{equation*}
y_1 + y_3 =0.
\end{equation*}
This contradicts the above equation.
\end{enumerate}

Consequently, if the equation (\ref{lb H}) has positive definite solution $Y$, it must satisfy $\frac{0.2}{\gamma_{\infty}^2}-\frac{0.5}{\delta^2}<0$, implying the condition $\gamma_{\infty} >\sqrt{0.4} \delta$.
\end{myremark}

\subsection{Mixed LQG and $H_{\infty}$ control problem}

After above derivations, we find that when we consider $H_{\infty}$ control, we intend to design a controller $K$ to minimize $\lVert \Gamma (sI-M)^{-1} H \rVert_{\infty}$, which depends on matrices $M$, $H$ and $\Gamma$, but these three matrices only depend on controller matrices $A_k$, $B_{k3}$ and $C_k$. Then we use physical realizability constraints to design other matrices $B_{k1}$ and $B_{k2}$ to guarantee the controller is also a quantum system, but these will affect the LQG index, which depends on $M$, $N$, $\Psi$, so further depends on all matrices of the controller. That is, the LQG problem and the $H_{\infty}$ problem are not independent.

According to the above analysis, we state the mixed LQG and $H_{\infty}$ coherent feedback control problem for linear quantum systems.

\begin{myproblem}\label{ori mixed}
The mixed LQG and $H_{\infty}$ coherent feedback control problem is to find a quantum controller $K$ of form (\ref{controller}) that minimizes LQG and $H_{\infty}$ performance indices simultaneously, while its coefficient matrices  satisfy the physical realizability constraints (\ref{our phy thm}).
\end{myproblem}

\begin{mylemma}\label{mixed cond}(Relaxed mixed problem \citep{BZL2015})
Given $\Theta_k$, $\gamma_l >0$ and $\gamma_{\infty} >0$, if there exist $A_k$, $B_{k1}$, $B_{k2}$, $B_{k3}$, $C_k$, $Q$, and symmetric matrices $P_L=P^{-1}$, $P_H$ such that physical realizability constraints (\ref{our phy thm}) and inequality constraints (\ref{nLMI LQG}) and (\ref{nLMI Hinf}) hold, where $P$ is the solution of equation (\ref{lya eq}), then the mixed LQG and $H_{\infty}$ coherent feedback control problem admits a coherent feedback controller $K$ of the form (\ref{controller}).
\end{mylemma}

\section{Algorithms for mixed LQG and $H_{\infty}$ coherent feedback control problem}\label{method}

In this section, the coherent feedback controllers for mixed LQG and $H_{\infty}$ problems are constructed by using two different methods, rank constrained LMI method and genetic algorithm (GA).

\subsection{Rank constrained LMI method}

In Lemma \ref{mixed cond} for the mixed problem, obviously constraints (\ref{nLMI LQG}) and (\ref{nLMI Hinf}) are non-linear matrix inequalities, and physical realizability conditions (\ref{our phy thm}) are non-convex constraints. Therefore, it is difficult to obtain the optimal solution by existing optimization algorithms. By referring to \citet{BZL2015,NJP2009,SGC1997}, we could translate these non-convex and non-linear constraints to linear ones.

Firstly, we redefine the original plant (\ref{plant}) to a \emph{modified plant} as follows
\begin{equation}\begin{split}
dx(t) &=Ax(t)dt+B_w d\tilde{w}_{cl}(t)+B_1 \beta_w (t)dt \\
&\phantom{=}+B_2 \beta_u (t)dt, \\
dy'(t) &=[b_{vk1}^T (t)\ b_{vk2}^T (t)\ y^T (t)]^T \\
&=Cx(t)dt+D_w d\tilde{w}_{cl}(t)+D\beta_w (t)dt, \\
dz_{\infty}(t) &=C_1 x(t)dt+D_{\infty}d\tilde{w}_{cl}(t)+D_{12}\beta_u (t)dt, \\
z_l (t) &=C_z (t)+D_z \beta_u (t),
\end{split}\end{equation}
where $B_w =[B_0\ B_1\ B_2\ 0]$, $C=[0\ 0\ C_2^T]^T$, $D=[0\ 0\ D_{12}^T]^T$, $D_{\infty}=[0\ 0\ D_{12}\ 0]$ and $D_w=\begin{bmatrix}
0 & 0 & I & 0 \\
0 & 0 & 0 & I \\
D_{20} & D_{21} & 0 & 0
\end{bmatrix}$.
Correspondingly, the \emph{modified controller} equations become the following form
\begin{equation}\begin{split}
d\xi (t) &=A_k \xi (t)dt+B_{wk}dy'(t), \\
\beta_u (t) &=C_k \xi (t)
\end{split}\end{equation}
with $B_{wk}=[B_{k1}\ B_{k2}\ B_{k3}]$, and the closed-loop system still has the same form as (\ref{cl-loop system}).

\begin{myassumption}
For simplicity we assume $P_H=P_L=P^{-1}$.
\end{myassumption}

We proceed to introduce matrix variables $\Xi$, $\Sigma$, $X$, $Y$, $Q \in \mathbb{R}^{n\times n}$, where $X$, $Y$ and $Q$ are symmetric. Then define the change of controller variables as follows
\begin{equation}\label{var change}\begin{split}
\hat{A} &\coloneqq \Xi A_k \Sigma^T+\Xi B_{wk}C X+Y B_2 C_k \Sigma^T+Y A X, \\
\hat{B} &\coloneqq \Xi B_{wk}, \\
\hat{C} &\coloneqq C_k \Sigma^T,
\end{split}\end{equation}
where $\Sigma \Xi^T=I-X Y$.

By using (\ref{var change}), LQG inequality constrains (\ref{nLMI LQG}) can be transformed to (\ref{LMI LQG}). Similarly, $H_{\infty}$ inequality constraints (\ref{nLMI Hinf}) become (\ref{LMI Hinf}). It is obvious that the following matrix inequalities are linear, so they can be easily solved by Matlab.
\begin{figure*}[!htb]
\normalsize
\setcounter{equation}{31}
{\small
\begin{equation}\label{LMI LQG}\begin{split}
\begin{bmatrix}
AX+XA^{T}+B_{2}\hat{C}+(B_{2}\hat{C})^{T} & \hat{A}^{T}+A & B_{w} \\
\hat{A}+A^{T} & A^{T}Y+YA+\hat{B}C+(\hat{B}C)^{T} & YB_{w}+\hat{B}D_{w} \\
B_{w}^{T} & (YB_{w}+\hat{B}D_{w})^{T} & -I
\end{bmatrix} &<0, \\
\begin{bmatrix}
X & I & (C_{z}X+D_{z}\hat{C})^{T} \\
I & Y & C_{z}^{T} \\
(C_{z}X+D_{z}\hat{C}) & C_{z} & Q
\end{bmatrix} &>0, \\
\mathrm{Tr}(Q) &<\gamma_l.
\end{split}
\end{equation}
\hrulefill
\begin{equation}
\begin{bmatrix}
AX+XA^{T}+B_{2}\hat{C}+(B_{2}\hat{C})^{T} & \hat{A}^{T}+A & \ast & \ast \\
\hat{A}+A^{T} & A^{T}Y+YA+\hat{B}C+(\hat{B}C)^{T} & \ast & \ast \\
B_{1}^{T} & (YB_{1}+\hat{B}D)^{T} & -\gamma_{\infty}I & \ast \\
C_{1}X+D_{12}\hat{C} & C_{1} & 0 & -\gamma_{\infty}I
\end{bmatrix}<0.  \label{LMI Hinf}
\end{equation}}
\setcounter{equation}{33}
\vspace*{4pt}
\end{figure*}

From (\ref{var change}), we can obtain $C_k=\hat{C} \Sigma^{-T}$, $B_{wk}=\Xi^{-1} \hat{B}$, and $A_k=\Xi^{-1} (\hat{A}-\Xi B_{wk} C X-Y B_2 C_k \Sigma^T-Y A X)\Sigma^{-T}$. After substituting $A_k$, $B_{wk}$ and $C_k$ into (\ref{our phy thm}) and introducing new variables $\tilde{\Xi}=\Xi J_{N_{\zeta}}$, $\tilde{A}_k=\Xi A_k$, $\tilde{B}_{ki}=\Xi B_{ki}$, $i=1,2,3$, physical realizability constraints (\ref{our phy thm}) become
\begin{subequations}\label{linear phy thm}\begin{align}
&(-\hat{A} \Sigma^{-T}+(\tilde{B}_{k3} C_2+Y A)X \Sigma^{-T}+Y B_2 C_k)\tilde{\Xi}^T \nonumber \\
&+\tilde{\Xi} (\hat{A} \Sigma^{-T}-(\tilde{B}_{k3} C_2+Y A)X \Sigma^{-T}-Y B_2 C_k)^T \nonumber \\
&+\sum_{i=1}^3 \tilde{B}_{ki} J_{n_{vki}/2} \tilde{B}_{ki}^T=0, \\
&\tilde{B}_{k1}-\tilde{\Xi} C_k^T J_{n_{vk1}/2}=0.
\end{align}\end{subequations}

We get the following result for the mixed LQG and $H_{\infty}$ coherent feedback control problem.

\begin{mylemma}
Given $\Theta_k$, $\gamma_l>0$ and $\gamma_{\infty}>0$, if there exist matrices $\hat{A}$, $\tilde{B}_{k1}$, $\tilde{B}_{k2}$, $\tilde{B}_{k3}$, $\hat{C}$, $X$, $Y$, $\tilde{\Xi}$, $\Sigma$, $\Xi$, $C_k$ such that the LMIs (\ref{LMI LQG}), (\ref{LMI Hinf}) and equality constraints (\ref{linear phy thm}) hold, then the mixed LQG and $H_{\infty}$ coherent feedback control problem admits a coherent feedback controller $K$ of the form (\ref{controller}).
\end{mylemma}

\begin{myalgorithm}(Rank constrained LMI method \citep{BZL2015})

Firstly, introduce 13 basic matrix variables: $M_1=\hat{A}$, $M_2=\tilde{B}_{k1}$, $M_3=\tilde{B}_{k2}$, $M_4=\tilde{B}_{k3}$, $M_5=\hat{C}$, $M_6=X$, $M_7=Y$, $M_8=\tilde{\Xi}$, $M_9=\Sigma$, $M_{10}=\Xi$, $M_{11}=C_k$, $M_{12}=\check{A}=\hat{A}\Sigma^{-T}$, $M_{13}=\check{X}=X \Sigma^{-T}$. And define 18 matrix lifting variables: $W_i=\tilde{B}_{ki} J_{N_{vki}}\ (i=1,2,3)$, $W_4=Y B_2$, $W_5=\tilde{B}_{k3} C_2 +Y A$, $W_6=\tilde{\Xi} C_k^T$, $W_7=\tilde{\Xi} \check{X}^T$, $W_8=\check{A} \tilde{\Xi}^T$, $W_9=Y X$, $W_{10}=W_4 W_6^T$, $W_{11}=W_5 W_7^T$, $W_{12}=W_1 \tilde{B}_{k1}^T$, $W_{13}=W_2 \tilde{B}_{k2}^T$, $W_{14}=W_3 \tilde{B}_{k3}^T$, $W_{15}=\Xi \Sigma^T=I-Y X$, $W_{16}=\check{A} \Sigma^T=\hat{A}$, $W_{17}=\check{X} \Sigma^T=X$, $W_{18}=C_k \Sigma^T=\hat{C}$.

By defining
\begin{equation}\begin{split}
V &=[I\ Z_{m_1,1}^T\ \cdots\ Z_{m_{13},1}^T\ Z_{w_1,1}^T\ \cdots\ Z_{w_{18},1}^T]^T \\
&=[I\ M_1^T\ \cdots\ M_{13}^T\ W_1^T\ \cdots\ W_{18}^T]^T,
\end{split}\end{equation}
we could let $Z$ be a $32n\times 32n$ symmetric matrix with $Z=V V^T$. It is obvious that $Z_{m_i,w_i}=Z_{m_i,1} (Z_{w_i,1})^T$.

Meanwhile, because of relations between these 31 variables, we require the matrix $Z$ to satisfy the following constraints
\begin{equation}\begin{split}
Z &\geq 0,\text{ \ \ \ } \\
Z_{0,0}-I_{n\times n} &=0,\text{ \ \ \ }%
Z_{w_{7},1}-Z_{m_{8},m_{13}}=0, \\
Z_{1,x_{6}}-Z_{m_{6},1} &=0,\text{ \ \ \ }%
Z_{w_{8},1}-Z_{m_{12},m_{8}}=0 \\
Z_{1,x_{7}}-Z_{m_{7},1} &=0,\text{ \ \ \ }%
Z_{w_{9},1}-Z_{m_{7},m_{6}}=0, \\
Z_{w_{1},1}-Z_{m_{2},1}J_{n_{vk1}/2} &=0,\text{ \ \ \ }%
Z_{w_{10},1}-Z_{w_{4},w_{6}}=0, \\
Z_{w_{2},1}-Z_{m_{3},1}J_{n_{vk2}/2} &=0,\text{ \ \ \ }%
Z_{w_{11},1}-Z_{w_{5},w_{7}}=0, \\
Z_{w_{3},1}-Z_{m_{4},1}J_{n_{vk3}/2} &=0,\text{ \ \ \ }%
Z_{w_{12},1}-Z_{w_{1},m_{2}}=0, \\
Z_{w_{4},1}-Z_{m_{7},1}B_{2} &=0,\text{ \ \ \ }%
Z_{w_{13},1}-Z_{w_{2},m_{3}}=0, \\
Z_{w_{5},1}-Z_{m_{4},1}C_{2}-Z_{m_{7},1}A &=0,\text{ \ \ \ }%
Z_{w_{14},1}-Z_{w_{3},m_{4}}=0, \\
Z_{w_{6},1}-Z_{m_{8},m_{11}} &=0,\text{ \ \ \ }%
Z_{w_{15},1}-Z_{m_{10},m_{9}}=0, \\
\text{ \ \ \ \ }Z_{w_{16},1}-Z_{m_{12},m_{9}} &=0,\text{ \ \ \ }Z_{w_{17},1}-Z_{m_{13},m_{9}}=0, \\
Z_{w_{18},1}-Z_{m_{11},m_{9}} &=0,\text{ \ \ \ }%
Z_{w_{15},1}-I+Z_{w_{9},1}=0, \\
\text{\ }Z_{m_{1},1}-Z_{w_{16},1} &=0,\text{ \ \ \ }%
Z_{m_{6},1}-Z_{w_{17},1}=0, \\
Z_{m_{8},1}-Z_{m_{10},1}J_{n_{\xi}/2} &=0,\text{ \ \ \ }%
Z_{m_{5},1}-Z_{w_{18},1}=0,
\end{split}\end{equation}
and moreover, $Z$ satisfies a rank $n$ constraint, i.e. $rank(Z) \leq n$.

Then, we use $Z_{m_1,1}$, $[Z_{m_2,1}\ Z_{m_3,1}\ Z_{m_4,1}]$, $Z_{m_5,1}$, $Z_{m_6,1}$, $Z_{m_7,1}$ to replace $\hat{A}$, $\hat{B}$, $\hat{C}$, $X$, $Y$ in LMI constraints (\ref{LMI LQG}) and (\ref{LMI Hinf}), and convert physical realizability conditions (\ref{linear phy thm}) to
\begin{subequations}\begin{align}
&-Z_{w_8,1}+Z_{w_8,1}^T+Z_{w_{11},1}-Z_{w_{11},1}^T+Z_{w_{10},1}-Z_{w_{10},1}^T \nonumber \\
&+Z_{w_{12},1}+Z_{w_{13},1}+Z_{w_{14},1}=0, \\
&Z_{m_2,1}-Z_{w_6,1} J_{N_{vk1}}=0.
\end{align}\end{subequations}

We have transformed the mixed problem to a rank constrained problem, which could be solved by using Toolbox: Yalmip \citep{Lof2004}, SeDuMi and LMIRank \citep{OHM2006}.
\end{myalgorithm}

\begin{myremark}
The above LMI-based approach solves a sub-optimal control problem for the mixed $LQG/H_\infty$ coherent feedback control. Once a feasible solution is found by implementing Algorithm 1, we then know that the LQG index is bounded by $\gamma_l$ from above, and {\it simultaneously}, the $H_\infty$ index is bounded by $\gamma_\infty$ from above.
\end{myremark}

\subsection{Genetic algorithm}

Genetic algorithm is a search heuristic that mimics the process of natural selection in the field of artificial intelligence. This heuristic (sometimes called metaheuristic) is routinely used to generate useful solutions to optimization and search problems. Genetic algorithm belongs to the larger class of evolutionary algorithms, which get solutions using techniques inspired by natural evolution, such as inheritance, mutation, selection, and crossover, etc. Genetic algorithm is a useful method for controller design, see e.g., \citet{CZ2003,NA2004,PA2004}. In the field of quantum control, genetic algorithm methods are applied to design quantum coherent feedback controllers, see e.g.,  \citet{ZLH2012} and \citet{HP2015}.

We briefly introduce the procedures of genetic algorithm as follows.

\begin{myalgorithm}(Genetic algorithm)

\begin{description}
\item[Step 1]: Initialization for the population (the first generation), by using random functions, and binary strings denote controller parameters we want to design;
\item[Step 2]: Transform binary strings to decimal numbers, and calculate the results of these parameters;
\item[Step 3]: After obtaining coefficient matrices of the controller, we restrict one of the LQG or $H_{\infty}$ indices in an interval, then minimize the other index (the fitness function in our problem). Since the lower bounds of these two indices can be calculated a priori, see Sections \ref{LQG analysis} and \ref{Hinf analysis}, the above-mentioned interval can always be found. By the above procedure we get the best individual and corresponding performance index in this generation;
\item[Step 4]: Perform the selection operation, for yielding new individuals;
\item[Step 5]: Perform the crossover operation, for yielding new individuals;
\item[Step 6]: Perform the mutation operation, for yielding new individuals;
\item[Step 7]: Back to Step 2, recalculate all parameters and corresponding best fitness function result for new generation;
\item[Step 8]: At the end of iterations, compare all best results of every generation, and get the optimal solution.
\end{description}
\end{myalgorithm}

\begin{myremark}
Algorithm 2 does not minimize both LQG and $H_\infty$ performance indices simultaneously. More specifically, as can be seen in Step 3, one of the indices is first fixed, then the other one is minimized. This procedure is repeated as can be seen from Step 7. Therefore, Algorithm 2 is an iterative minimization algorithm.
\end{myremark}

In our problem, because the coherent feedback controller $K$ to be designed is a quantum system, it can be described by the $(S,L,H)$ language introduced in Subsection \ref{subsec:system}. With this, physical realizability conditions are naturally satisfied. As a result we apply the GA to find $K$ by minimizing the LQG and $H_\infty$ performance indices directly.

\section{Numerical simulations and comparisons}\label{examples}

In this section, we provide two examples to illustrate the methods proposed in the previous section.

\subsection{Numerical simulations}

\emph{Example 1:} This example is taken from Section \Rmnum{7} of \citet{JNP2008}. The plant is an optical cavity resonantly coupled to three optical channels.

The dynamics of this optical cavity system can be described by following equations
\begin{equation}\label{example1_plant}\begin{split}
dx\left( t\right)  &=-\frac{\gamma }{2}
\begin{bmatrix}
1 & 0 \\
0 & 1\\
\end{bmatrix}
x\left( t\right) dt-\sqrt{\kappa _{1}}%
\begin{bmatrix}
1 & 0 \\
0 & 1%
\end{bmatrix}dv(t)\\
&\phantom{22}-\sqrt{\kappa _{2}}
\begin{bmatrix}
1 & 0 \\
0 & 1%
\end{bmatrix}
dw\left( t\right) -\sqrt{\kappa _{3}}%
\begin{bmatrix}
1 & 0 \\
0 & 1\\
\end{bmatrix}
du(t),  \\
dy(t)&=\sqrt{\kappa _{2}}
\begin{bmatrix}
1 & 0 \\
0 & 1 \\
\end{bmatrix}x(t)dt+
\begin{bmatrix}
1 & 0 \\
0 & 1\\
\end{bmatrix}dw(t), \\
dz_{\infty}(t)&=\sqrt{\kappa _{3}}
\begin{bmatrix}
1 & 0 \\
0 & 1\\
\end{bmatrix}x(t)dt+
\begin{bmatrix}
1 & 0 \\
0 & 1\\
\end{bmatrix}du(t),    \\
z_{l}(t)&=
\begin{bmatrix}
1 & 0 \\
0 & 1\\
\end{bmatrix}x(t)+
\begin{bmatrix}
1 & 0 \\
0 & 1\\
\end{bmatrix}\beta _{u}(t)
\end{split}\end{equation}
with parameters $\gamma =\kappa _{1}+\kappa _{2}+\kappa _{3}$, $\kappa _{1}=2.6$, $\kappa _{2}=\kappa _{3}=0.2$. In this example, $v(t)$ is quantum white noise, while $w(t)$ is a sum of quantum white noise and $L_2$ disturbance (See Assumption 1 for details). Therefore, there are two types of noises in this system. LQG control is used to suppress the influence of quantum white noise, while $H_\infty$ control is used to attenuate the $L_2$ disturbance.

\emph{Example 2:} In this example, we choose a DPA as our plant. For more details about DPA, one may refer to \citet{Leo2003}. The QSDEs of DPA are
\begin{equation}\label{example2_plant}
\begin{split}
dx\left( t\right)  &=-\frac{1}{2}
\begin{bmatrix}
\gamma-\epsilon & 0 \\
0 & \gamma+\epsilon\\
\end{bmatrix}
x\left( t\right) dt-\sqrt{\kappa _{3}}%
\begin{bmatrix}
1 & 0 \\
0 & 1%
\end{bmatrix}dv(t)\\
&\phantom{22}-\sqrt{\kappa _{1}}
\begin{bmatrix}
1 & 0 \\
0 & 1%
\end{bmatrix}
dw\left( t\right)-\sqrt{\kappa _{2}}%
\begin{bmatrix}
1 & 0 \\
0 & 1\\
\end{bmatrix}
du(t), \\
dy(t)&=\sqrt{\kappa _{3}}
\begin{bmatrix}
1 & 0 \\
0 & 1 \\
\end{bmatrix}x(t)dt+
\begin{bmatrix}
1 & 0 \\
0 & 1\\
\end{bmatrix}dv(t), \\
dz_{\infty}(t)&=\sqrt{\kappa _{2}}
\begin{bmatrix}
1 & 0 \\
0 & 1
\end{bmatrix}x(t)dt+
\begin{bmatrix}
1 & 0 \\
0 & 1\\
\end{bmatrix}du(t),    \\
z_{l}(t)&=
\begin{bmatrix}
1 & 0 \\
0 & 1\\
\end{bmatrix}x(t)+
\begin{bmatrix}
1 & 0 \\
0 & 1\\
\end{bmatrix}\beta _{u}(t)
\end{split}\end{equation}
with parameters $\gamma =\kappa _{1}+\kappa _{2}+\kappa _{3}$, $\kappa_{1}=\kappa_{2}=0.2$, $\kappa _{3}=0.5$, $\epsilon=0.01$.

According to Theorem \ref{lower bound LQG}, it is easy to find that lower bounds of the LQG index for both two examples are 1. Firstly, we only focus on the LQG performance index, and design two different types of controllers to minimize it by using genetic algorithm. The results are shown in Table~\ref{res one index LQG}. For each case, we list two values obtained.

\begin{table}[!htbp]
\tbl{Optimization results only for LQG index.}
{\small
\begin{tabular}{c|c|c}
\hline
plant & controller & $J_{\infty}$ (LQG index) \\
\hline
\multirow{4}{*}{Cavity} & \multirow{2}{*}{Passive Controller} & 1.0005 \\
\cline{3-3}
& & 1.0000 \\
\cline{2-3}
& \multirow{2}{*}{Non-passive Controller} & 1.0006 \\
\cline{3-3}
& & 1.0003 \\
\hline
\multirow{4}{*}{DPA} & \multirow{2}{*}{Passive Controller} & 1.0003 \\
\cline{3-3}
& & 1.0000 \\
\cline{2-3}
& \multirow{2}{*}{Non-passive Controller} & 1.0002 \\
\cline{3-3}
& & 1.0000 \\
\hline
\end{tabular}}
\label{res one index LQG}
\end{table}

\begin{myremark}
$J_\infty$ in Table 1 is the LQG performance index defined in Eq. (\ref{eq:lqg}).  In Theorem \ref{lower bound LQG}, a lower bound for $J_\infty$ is proposed. This lower bound is obtained when both the plant and the controller are in the ground state, as stated in Theorem \ref{lower bound LQG}. In Table 1 above there are two systems, namely the optical cavity and DPA. For both of them, the lower bound in Theorem \ref{lower bound LQG} satisfies  $d_2 =0$ and $c_1=c_3=1$. Therefore, $J_\infty\geq 1$. From Table 1 we can see that our genetic algorithm finds controllers that yield the LQG performance which is almost optimal. And in this case, as guaranteed by Theorem \ref{lower bound LQG}, both the  plant and the controller are almost in the ground state.
\end{myremark}

Secondly, similarly to the LQG case, we only focus on the $H_{\infty}$ index and design controllers to minimize the objective, getting the following Table~\ref{res one index Hinf}. For each case, we list two values obtained.

\begin{table}[!htbp]
\tbl{Optimization results only for $H_{\infty}$ index.}
{\small
\begin{tabular}{c|c|c}
\hline
plant & controller & $\Vert G_{\beta _{w}\rightarrow z_{\infty}}\Vert _{\infty}$ ($H_{\infty}$ index) \\
\hline
\multirow{4}{*}{Cavity} & \multirow{2}{*}{Passive Controller} & 0.0134 \\
\cline{3-3}
& & 0.0050 \\
\cline{2-3}
& \multirow{2}{*}{Non-passive Controller} & 0.0196 \\
\cline{3-3}
& & 0.0075 \\
\hline
\multirow{4}{*}{DPA} & \multirow{2}{*}{Passive Controller} & 0.0070 \\
\cline{3-3}
& & 0.0044 \\
\cline{2-3}
& \multirow{2}{*}{Non-passive Controller} & 0.0057 \\
\cline{3-3}
& & 0.0045 \\
\hline
\end{tabular}}
\label{res one index Hinf}
\end{table}

\begin{myremark}
Table 2 is for $H_\infty$ performance index. For the cavity case, actually it can be proved analytically that the  $H_\infty$ performance index can be made arbitrarily close to zero. On the other hand, by Remark \ref{lower bound Hinf}, $H_\infty$ index has a lower bound $\sqrt{0.4}\delta$. However, for the DPA studied in this example, $\delta=0$, that is, the lower bound for $H_\infty$ index is also zero. The simulation results in Table 2 confirmed this observation.
\end{myremark}

From above results we could see, if we only consider one performance index, either LQG index or $H_{\infty}$ index,  there are no significant differences between passive controllers and non-passive controllers, both of which can lead to a performance index close to the minimum.

Then we proceed to use these two methods to do simulations for the mixed problem, to see whether we could succeed to make two performance indices close to the minima simultaneously, and find which method is better. The results are shown in Table~\ref{results LMI} and Table~\ref{results GA}, respectively.

\begin{table}[htbp]
\tbl{Optimization results by rank constrained LMI method.}
{\small
\begin{tabular}{c|c|c|c|c}
\hline
\multirow{2}{*}{plant}  & \multicolumn{2}{|c|}{constraints} & \multicolumn{2}{c}{results}\\
\cline{2-5}
& $\gamma_\infty$ & $\gamma_l$ &$\Vert G_{\beta _{w}\rightarrow z_{\infty}}\Vert _{\infty}$ ($H_{\infty}$ index) & $J_{\infty}$ (LQG index) \\
\hline
\multirow{5}{*}{\tabincell{c}{Cavity\\ ($\gamma =\kappa _{1}+\kappa _{2}+\kappa _{3}$,\\ $\kappa _{1}=2.6$,\\$\kappa _{2}=\kappa _{3}=0.2$.)}} & 0.1 & 2.5 & 0.039900 & 1.014555 \\
\cline{2-5}
& 0.1 & N/A & 0.058805 & 1.039487 \\
\cline{2-5}
& N/A & 2.5 & 0.134558 & 1.000577 \\
\cline{2-5}
& N/A & 3 & 0.423970 & 1.379587 \\
\cline{2-5}
& 2.8 & 3 & 0.444119 & 1.270835 \\
\hline
\multirow{3}{*}{\tabincell{c}{DPA\\($\gamma=\kappa_1+\kappa_2+\kappa_3$,\\ $\kappa_1=\kappa_2=0.2$,\\$\kappa_3=0.5, \epsilon=0.01$.)}} & 0.3 & 2.5 & 0.172385 & 1.175277 \\
\cline{2-5}
& 0.5 & 3 & 0.447274 & 1.080976 \\
\cline{2-5}
& N/A & 3 & 0.468007 & 1.149859 \\
\cline{2-5}
& 1 & 5 & 0.647468 & 1.374547 \\
\hline
\end{tabular}}
\label{results LMI}
\end{table}

\begin{table}[htbp]
\tbl{Optimization results by genetic algorithm.}
{\small
\begin{tabular}{c|c|c|c}
\hline
\multirow{2}{*}{plant}  & \multirow{2}{*}{controller} & \multicolumn{2}{c}{results}\\
\cline{3-4}
& & {$\Vert G_{\beta _{w}\rightarrow z_{\infty}}\Vert _{\infty}$ ($H_{\infty}$ index)} & {$J_{\infty}$ (LQG index)} \\
\hline
\multirow{6}{*}{\tabincell{c}{Cavity\\ ($\gamma =\kappa _{1}+\kappa _{2}+\kappa _{3}$,\\ $\kappa _{1}=2.6$,\\$\kappa _{2}=\kappa _{3}=0.2$.)}} & \multirow{3}{*}{\tabincell{c}{Passive Controller}} & 0.003574 & 1.008917 \\
\cline{3-4}
& & 0.078977 & 1.000619 \\
\cline{3-4}
& & 0.146066 & 1.000009 \\
\cline{2-4}
& \multirow{3}{*}{\tabincell{c}{Non-passive Controller}} & 0.071270 & 1.009303 \\
\cline{3-4}
& & 0.089383 & 1.002099 \\
\cline{3-4}
& & 0.123066 & 1.000283 \\
\hline
\multirow{6}{*}{\tabincell{c}{DPA\\($\gamma=\kappa_1+\kappa_2+\kappa_3$,\\ $\kappa_1=\kappa_2=0.2$,\\$\kappa_3=0.5, \epsilon=0.01$.)}} & \multirow{2}{*}{\tabincell{c}{Passive Controller}} & 0.428312 & 1.004787 \\
\cline{3-4}
& & 0.449534 & 1.000124 \\
\cline{2-4}
& \multirow{2}{*}{\tabincell{c}{Non-passive Controller}} & 0.364979 & 1.009691 \\
\cline{3-4}
& & 0.387734 & 1.007164 \\
\cline{2-4}
& \multirow{2}{*}{\tabincell{c}{Passive Controller \\ $+$ Direct coupling}} & 0.039183 & 1.000079 \\
\cline{3-4}
& & 0.042960 & 1.000002 \\
\hline
\end{tabular}}
\label{results GA}
\end{table}

\subsection{Comparisons of results}

After getting numerical results shown in Table~\ref{results LMI} and Table~\ref{results GA}, and doing comparisons with other literatures in coherent optimal control for linear quantum systems, we state the advantages of Algorithm 2:
\begin{enumerate}
\item Instead of single LQG or $H_{\infty}$ optimal control for linear quantum systems, Algorithm 2 deals with the {\it mixed} LQG and $H_{\infty}$ problem.

\item Algorithm 1 relaxes two performance indices by introducing $\gamma_l$ and $\gamma_\infty$. When they are small, it will be quite difficult to solve the problem by Algorithm 1.  But  Algorithm 2 is able to minimize the  two performance indices directly.

\item The solution of the differential evolution algorithm in \citet{HP2015} involves a complex algebraic Riccati equation, but all parameters of our Algorithm 2 are real. It might be easier to be solved by current computer software such as Matlab.

\item The numerical results show that there seems to be a trend between these two indices, that sometimes one increases, while another decreases.

\item For a passive system (e.g. cavity), both the passive controller and the non-passive controller could let LQG and $H_{\infty}$ indices go to the minima simultaneously, Table~\ref{results GA}.

\item For a non-passive system (e.g. DPA), neither the passive controller nor the non-passive controller can let these two indices go to the minima simultaneously, but when a direct coupling is added between the plant and the controller, we could use genetic algorithm to design a passive controller to minimize these two indices simultaneously, which is not achieved using rank constrained LMI method, Table~\ref{results GA}.

\item Actually, rank constrained LMI method could not be used to design specific passive controllers, or non-passive controllers, while this can be easily achieved using genetic algorithm, by setting all ``plus'' terms equal to 0.

\item Finally, from numerical simulations, genetic algorithm often provides better results than the rank constrained LMI method.
\end{enumerate}

\section{Conclusion}\label{conclusion}

In this paper, we have studied the mixed LQG and $H_{\infty}$ coherent feedback control problem.  Two algorithms,  rank constrained LMI method and a genetic algorithm-based method, have been proposed.  Two examples are used to illustrate the effectiveness of these two methods, and also verify the superiority of genetic algorithm by numerical results.



\section*{Funding}
This research is supported in part by National Natural Science Foundation of China (NSFC) grants (Nos. 61374057) and Hong Kong RGC grants (Nos. 531213 and 15206915).



\begin{thebibliography}{99}

\bibitem[Bian et al.(2015)Bian, Zhang, \& Lee]{BZL2015}
Bian, C., Zhang, G.,  \& Lee, H-W. (2015). LQG/$H_\infty$ control of linear quantum stochastic systems. {\em{Chinese Control Conference}}, Hangzhou, China, 28--30 July 2015 (pp. 8303--8308).

\bibitem[Bouten et al.(2007)Bouten, Handel, \& James]{BHJ2007}
Bouten, L., Handel, R., \& James, M. (2007). An introduction to quantum filtering. {\em{SIAM J. Control Optim.}}, {\em{46\,}}(6), 2199--2241.

\bibitem[Campos-Delgado \& Zhou(2003)]{CZ2003}
Campos-Delgado, D., \& Zhou, K. (2003). A parametric optimization approach to $H_{\infty}$ and $H_2$ strong stabilization. {\em{Automatica}}, {\em{39\,}}(7), 1205--1211.

\bibitem[Doherty \& Jacobs(1999)]{DJ1999}
Doherty, A., \& Jacobs, K. (1999). Feedback control of quantum systems using continuous state estimation. {\em{Phys. Rev. A}}, {\em{60\,}}, 2700--2711.

\bibitem[Doherty et al.(2000)Doherty, Habib, Jacobs, Mabuchi, \& Tan]{DHJ2000}
Doherty, A., Habib, S., Jacobs, K., Mabuchi, H., \& Tan, S. (2000). Quantum feedback control and classical control theory. {\em{Phys. Rev. A}}, {\em{62\,}}, 012105.


\bibitem[Doyle et al.(1994)Doyle, Zhou, Glover, \& Bodenheimer]{DZG1994}
Doyle, J., Zhou, K., Glover, K., \& Bodenheimer, B. (1994). Mixed $H_2$ and $H_{\infty}$ performance objectives. II. optimal control. {\em{IEEE Transactions on Automatic Control}}, {\em{39\,}}(8), 1575--1587.

\bibitem[Hamerly \& Mabuchi(2013)]{HM2013}
Hamerly, R., \& Mabuchi, H. (2013). Coherent controllers for optical-feedback cooling of quantum oscillators. {\em{Phys. Rev. A}}, {\em{87\,}}, 013815.

\bibitem[Hudson \& Parthasarathy(1984)]{HP1984}
Hudson, R., \& Parthasarathy, K. (1984). Quantum Ito's formula and stochastic evolutions. {\em{Commun. Math. Phys.}}, {\em{93\,}}, 301--323.

\bibitem[Harno \& Petersen(2015)]{HP2015}
Harno, H., \& Petersen, I. (2015). Synthesis of linear coherent quantum control systems using a differential evolution algorithm. {\em{IEEE Transactions on Automatic Control}}, {\em{60\,}}(3), 799--805.

\bibitem[James et al.(2008)James, Nurdin, \& Petersen]{JNP2008}
James, M., Nurdin, H., \& Petersen, I. (2008). $H_{\infty}$ control of linear quantum stochastic systems. {\em{IEEE Transactions on Automatic Control}}, {\em{53\,}}(8), 1787--1803.

\bibitem[Khargonekar \& Rotea(1991)]{KR1991}
Khargonekar, P., \& Rotea, M. (1991). Mixed $H_2$/$H_{\infty}$ control: a convex optimization approach. {\em{IEEE Transactions on Automatic Control}}, {\em{36\,}}(7), 824--837.

\bibitem[Leonhardt(2003)]{Leo2003}
Leonhardt, U. (2003). Quantum physics of simple optical instruments. {\em{Reports on Progress in Physics}}, {\em{66\,}}(7), 1207.

\bibitem[Lofberg(2004)]{Lof2004}
Lofberg, J. (2004). Yalmip: a toolbox for modeling and optimization in matlab. {\em{2004 IEEE International Symposium on Computer Aided Control Systems Design}}, Taipei, Taiwan, 2--4 September 2004 (pp. 284--289).

\bibitem[Merzbacher(1998)]{Mer1998}
Merzbacher, E. (1998). {\em{Quantum Mechanics}}(3rd ed.). New York: Wiley.

\bibitem[Michalewicz et al.(1992)Michalewicz, Janikow, \& Krawczyk]{MJK1992}
Michalewicz, Z., Janikow, C., \& Krawczyk, J. (1992). A modified genetic algorithm for optimal control problems. {\em{Comput. Math. Appl.}}, {\em{23\,}}(12), 83--94.

\bibitem[Neumann \& Araujo(2004)]{NA2004}
Neumann, D., \& Araujo, H. de (2004). Mixed $H_2$/$H_{\infty}$ control for uncertain systems under pole placement constraints using genetic algorithms and LMIs. {\em{Proceedings of the 2004 IEEE International Symposium on Intelligent Control}}, Taipei, Taiwan, 2--4 September 2004 (pp. 460--465).

\bibitem[Nurdin et al.(2009)Nurdin, James, \& Petersen]{NJP2009}
Nurdin, H., James, M., \& Petersen, I. (2009). Coherent quantum LQG control. {\em{Automatica}}, {\em{45\,}}(8), 1837--1846.

\bibitem[Orsi et al.(2006)Orsi, Helmke, \& Moore]{OHM2006}
Orsi, R., Helmke, U., \& Moore, J. (2006). A newton-like method for solving rank constrained linear matrix inequalities. {\em{Automatica}}, {\em{42\,}}(11), 1875--1882.

\bibitem[Pereira \& Araujo(2004)]{PA2004}
Pereira, G., \& Araujo, H. de (2004). Robust output feedback controller design via genetic algorithms and LMIs: the mixed $H_2$/$H_{\infty}$ problem. {\em{Proceedings of the 2004 American Control Conference}}, Boston, USA, 30 June--2 July 2004 (pp. 3309--3314).

\bibitem[Petersen(2013)]{Pet2013}
Petersen, I. (2013). Notes on coherent feedback control for linear quantum systems. {\em{2013 Australian Control Conference}}, Perth, Australia, 4--5 November 2013 (pp. 319--324).

\bibitem[Qiu et al.(2015)Qiu, Shi, Yao, G. Xu, \& B. Xu]{QSY2015}
Qiu, L., Shi, Y., Yao, F., Xu, G., \& Xu, B. (2015). Network-based robust $H_2$/$H_{\infty}$ control for linear systems with two-channel random packet dropouts and time delays. {\em{IEEE Transactions on Cybernetics}}, {\em{45\,}}(8), 1450--1462.

\bibitem[Scherer et al.(1997)Scherer, Gahinet, \& Chilali]{SGC1997}
Scherer, C., Gahinet, P., \& Chilali, M. (1997). Multiobjective output-feedback control via LMI optimization. {\em{IEEE Transactions on Automatic Control}}, {\em{42\,}}(7), 896--911.

\bibitem[Wang \& James(2015)]{WJ2015}
Wang, S., \& James, M. (2015). Quantum feedback control of linear stochastic systems with feedback-loop time delays. {\em{Automatica}}, {\em{52\,}}, 277--282.

\bibitem[Wiseman \& Milburn(2010)]{WM2010}
Wiseman, H., \& Milburn, G. (2010). {\em{Quantum measurement and control}}. New York: Cambridge University Press.

\bibitem[Wang et al.(2013)Wang, Nurdin, Zhang, \& James]{WNZ2013}
Wang, S., Nurdin, H., Zhang, G., \& James, M. (2013). Quantum optical realization of classical linear stochastic systems. {\em{Automatica}}, {\em{49\,}}(10), 3090--3096.

\bibitem[Zhang \& James(2011)]{ZJ2011}
Zhang, G., \& James, M. (2011). Direct and indirect couplings in coherent feedback control of linear quantum systems. {\em{IEEE Transactions on Automatic Control}}, {\em{56\,}}(7), 1535--1550.

\bibitem[Zhang \& James(2012)]{ZJ2012}
Zhang, G., \& James, M. (2012). Quantum feedback networks and control: A brief survey. {\em{Chinese Science Bulletin}}, {\em{57\,}}, 2200--2214.

\bibitem[Zhang et al.(2012)G. Zhang, Lee, Huang, \& H. Zhang]{ZLH2012}
Zhang, G., Lee, H. W. J., Huang, B., \& Zhang, H. (2012). Coherent feedback control of linear quantum optical systems via squeezing and phase shift. {\em{Siam Journal Control and Optimization}}, {\em{50\,}}(4), 2130--2150.

\bibitem[Zhou et al.(1996)Zhou, Doyle, \& Glover]{ZDG1996}
Zhou, K., Doyle, J., \& Glover, K. (1996). {\em{Robust and optimal control}}. Prentice-Hall, Inc..

\bibitem[Zhou et al.(1994)Zhou, Glover, Bodenheimer, \& Doyle]{ZGB1994}
Zhou, K., Glover, K., Bodenheimer, B., \& Doyle, J. (1994). Mixed $H_2$ and $H_{\infty}$ performance objectives. I. robust performance analysis. {\em{IEEE Transactions on Automatic Control}}, {\em{39\,}}(8), 1564--1574.

\end{thebibliography}
\end{document}